\documentclass[sigconf]{acmart}

\usepackage[noend,linesnumbered,lined,ruled]{algorithm2e}
\usepackage{amsthm}
\usepackage{enumerate}
\usepackage{scalerel}
\usepackage{array}
\usepackage{subfigure}
\usepackage{bbm}
\usepackage{dsfont}

\newtheorem{theorem}{Theorem}

\newtheorem{definition}{Definition}

\DeclareMathOperator*{\argmin}{\arg\min}

\AtBeginDocument{%
  }

\setcopyright{acmlicensed}
\copyrightyear{2025}
\acmYear{2025}
\acmDOI{XXXXXXX.XXXXXXX}
\acmConference[RAGE 2025]{The 4th Real-time And intelliGent Edge computing workshop}{May 06,
  2025}{Irvine, CA, USA}
\acmISBN{978-1-4503-XXXX-X/2018/06}

\begin{document}

\title{Local Ratio based Real-time Job Offloading and Resource Allocation in Mobile Edge Computing
}


\author{Chuanchao Gao, Arvind Easwaran}
\email{gaoc0008@e.ntu.edu.sg, arvinde@ntu.edu.sg}
\affiliation{%
  \institution{
    \textit{College of Computing and Data Science} \\
    \textit{Energy Research Institute @ NTU, Interdisciplinary Graduate Programme}\\
    \textit{Nanyang Technological University}
  }
  \country{Singapore}
}




\begin{abstract}
Mobile Edge Computing (MEC) has emerged as a promising paradigm enabling vehicles to handle computation-intensive and time-sensitive applications for intelligent transportation. Due to the limited resources in MEC, effective resource management is crucial for improving system performance. While existing studies mostly focus on the job offloading problem and assume that job resource demands are fixed and given apriori, the joint consideration of job offloading (selecting the edge server for each job) and resource allocation (determining the bandwidth and computation resources for offloading and processing) remains underexplored. This paper addresses the joint problem for deadline-constrained jobs in MEC with both communication and computation resource constraints, aiming to maximize the total utility gained from jobs. 
To tackle this problem, we propose an approximation algorithm, $\mathtt{IDAssign}$, with an approximation bound of $\frac{1}{6}$, and experimentally evaluate the performance of $\mathtt{IDAssign}$ by comparing it to state-of-the-art heuristics using a real-world taxi trace and object detection applications.
\end{abstract}




\maketitle

\addtolength{\dbltextfloatsep}{-2mm}
\addtolength{\abovecaptionskip}{-3mm}
\addtolength{\textfloatsep}{-3mm}
\addtolength{\floatsep}{-2mm}

\section{Introduction}
Internet of Vehicles (IoV), a derivative technology of the Internet of Things (IoT), plays a pivotal role in realizing intelligent transportation systems by connecting vehicles and infrastructure to enable real-time service support.
IoV utilizes advanced wireless communication technologies such as ultra-reliable low-latency communication of 5G to facilitate intelligent applications like smart traffic control, autonomous driving, and real-time navigation. These applications are generally computation-intensive and highly sensitive to latency, posing significant challenges for onboard vehicular systems with limited computing capabilities.

Mobile Edge Computing (MEC) has emerged as a promising paradigm enabling vehicles to support computation-intensive and time-critical applications in IoV. In MEC, vehicles can offload jobs to roadside edge servers (ES) that provide communication (bandwidth) and computation resources for instant processing. Deploying ES close to vehicles significantly reduces communication latency compared to traditional cloud computing, enabling prompt responses to vehicular requests. However, the bandwidth and computation resources of ES are limited, necessitating efficient strategies for \textit{job offloading} (selecting the ES for each job) and \textit{resource allocation} (determining the bandwidth and computation resource allocations for offloading and processing) in MEC.

Jointly optimizing job offloading and resource allocation enables MEC systems to dynamically allocate resources based on availability, improving overall resource utilization. For instance, an MEC system with abundant computational resources but limited bandwidth can allocate more computing power while reducing bandwidth per job, thereby accommodating more jobs while maintaining the required quality of service. While numerous studies have investigated the Deadline-constrained job Offloading Problem (DOP) in MEC under the assumption of fixed resource allocation per job \cite{li2022maximizing, li2023throughput, ma2022mobility, xu2024age, yang2024a}, the joint optimization of job offloading and resource allocation \cite{fan2023joint, hu2020dynamic, qi2024joint, li2019cooperative, gao2022deadline, vu2021optimal, wang2019joint}, referred to as the Deadline-constrained job Offloading and resource Allocation Problem (DOAP), remains underexplored.

DOP reduces to the two-dimensional Generalized Assignment Problem (2DGAP) when considering both bandwidth and computational resource contention (e.g., \cite{li2022maximizing}) and to the Generalized Assignment Problem (GAP, known to be NP-hard \cite{fleischer2006tight}) when only computational constraints are considered (e.g., \cite{li2022maximizing, li2023throughput, ma2022mobility, yang2024a, xu2024age}). DOAP is a more generalized version of DOP that additionally considers resource allocation, making it inherently NP-hard as well. Obtaining an exact solution for DOP or DOAP (e.g., \cite{vu2021optimal, fan2023joint}) requires exponential time. Consequently, most existing studies focus on polynomial-time heuristic algorithms. Among these, approximation algorithms stand out as a special class of heuristics that provide worst-case performance guarantees while maintaining low computational complexity. Prior studies that propose approximation algorithms (e.g., \cite{li2022maximizing, li2023throughput, ma2022mobility, yang2024a, xu2024age, wang2019joint}) primarily focus on the DOP variant, formulating their problems as GAP by considering only computational constraints with fixed resource allocations. Although Wang \emph{et al.} \cite{wang2019joint} addressed a DOAP and proposed an approximation algorithm with a probabilistic guarantee, their approach predefines resource allocation by manually specifying the offloading and processing times for each job. As a result, their approximation algorithm effectively applies only to DOP rather than the more general DOAP. To the best of our knowledge,  no approximation algorithm with constant approximation ratio has been proposed for DOAP.

DOAP constraints can be categorized into two types. The first type ensures that allocated resources and offloading decisions satisfy job-specific requirements, such as deadlines and ES accessibility. The second type enforces system-wide resource constraints, ensuring that the total allocated resources at each ES do not exceed its capacity. We define an \textit{assignment instance} of a job as a combination of (i) an offloading ES and (ii) allocated bandwidth and computational resources. Given an assignment instance, we can immediately verify whether it satisfies all first-type constraints and determine its contribution to the objective function. By focusing only on feasible assignment instances (i.e., those satisfying the first-type constraints), we only need to focus on the system-wide resource constraints in DOAP. Differences in the mathematical formulation of DOAP across various MEC applications primarily stem from variations in the first-type constraints and optimization objectives. Modeling DOAP based on assignment instances provides a uniform framework applicable to various scenarios.

In this paper, we study a DOAP in MEC with both bandwidth and computation resource constraints under the IoV scenario, aiming to maximize the total utility gained from offloading jobs; we refer to this specific DOAP as problem $\mathbf{P}$. We first formulate $\mathbf{P}$ as an Integer Linear Programming (ILP) problem by enumerating all possible assignment instances. Then, we introduce an approximation algorithm for $\mathbf{P}$, denoted as $\mathtt{IDAssign}$, based on a local ratio framework \cite{bar2004local}. We prove that $\mathtt{IDAssign}$ is a $\frac{1}{6}$-approximation algorithm for $\mathbf{P}$, i.e., the total utility obtained by $\mathtt{IDAssign}$ is no less than $\frac{1}{6}$ of the optimal utility of $\mathbf{P}$. Since $\mathbf{P}$ is formulated with the uniform framework, the approximation ratio of $\mathtt{IDAssign}$ is applicable across various DOAP applications. Finally, we evaluate $\mathtt{IDAssign}$ using profiled data from object detection applications and real-world taxi traces \cite{taxi-data}. While providing a theoretical guarantee, $\mathtt{IDAssign}$ demonstrates performance comparable to all evaluated baseline algorithms in the literature \cite{hu2020dynamic, qi2024joint, gao2022deadline, li2019cooperative, fan2023joint}.

\section{System Model and Problem Formulation}\label{sec:system}
In this section, we first introduce the MEC architecture, followed by a detailed formulation of DOAP in the presented MEC framework using its generalized form.

\subsection{System Architecture}
\label{sec:system_arch}
As illustrated in Fig.\ref{fig:system-vec}, we consider an MEC comprising a set of roadside ES providing services to moving vehicles. The notations used in the paper are summarized in Table\ref{table:notation}. Let $\mathcal{M}$ denote the set of $M$ ES, with $m_k \in \mathcal{M}$ representing the $k$-th ES. Each ES includes an access point (providing bandwidth for job offloading) and a server (providing computation resource for job processing). The computation resource capacity of each ES is measured in computing units, which depend on the specific server hardware configuration. For example, using NVIDIA’s multi-instance GPU technology \cite{nvidia-mig}, a GPU-enabled ES can partition its computation resources into multiple computing units, each comprising a predefined amount of memory, cache, and compute cores. Let $C_k$ be the number of computing units available at $m_k$, and let $C = \max\{C_k \mid m_k \in \mathcal{M}\}$.

\begin{figure}
    \centering
    \subfigure[]{\includegraphics[width=0.6\linewidth, page=1]{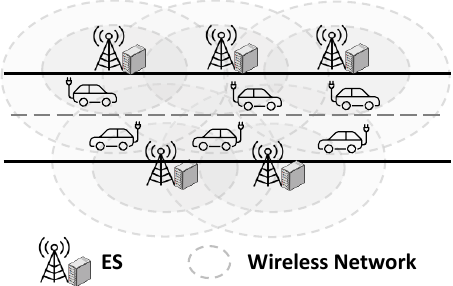} \label{fig:system-vec}} 
    \subfigure[]{\includegraphics[width=0.36\linewidth, page=2]{figures/system.pdf} \label{fig:network-ring}}
    \caption{(a) An example of mobile edge computing for IoV; (b) Partition the communication range of a wireless network into three network rings;}
    \Description{system figures}
\end{figure}

We consider a 5G network in MEC, where the Orthogonal Frequency Division Multiple Access (OFDMA) is used as the network access method for vehicles. In an OFDMA-enabled network, the smallest bandwidth allocation unit available to users is referred to as a resource unit or bandwidth unit (BU) \cite{ofdma}. Each user may be assigned one or more BU for data offloading. The size (or bandwidth) of a BU is measured in megahertz (MHz) and is determined by the specific configuration of the 5G network. Let $B_k$ denote the total number of BUs available at $m_k$, with $\beta_k$ representing the size of each BU. Define $B = \max\{B_k \mid m_k \in \mathcal{M}\}$. The wireless network of each $m_k$ has a limited effective communication range, and a vehicle can only offload a job to $m_k$ when it is within this range. Since the channel gain in wireless networks diminishes as the distance between a vehicle and an ES increases \cite{zhu2018task}, we divide each ES’s effective range into several \textit{network rings} such that the channel gain remains unchanged (from a practical viewpoint) within each network ring (Fig.~\ref{fig:network-ring}). Let $m_{kr}$ be the $r$-th network ring of $m_k$.

Let $\mathcal{N}$ represent the set of $N$ jobs to be processed. The $j$-th job, $n_j \in \mathcal{N}$, has an input data size $\theta_j$ (measured in megabytes, MB) and a deadline $\Delta_j$ (measured in seconds). When a job $n_j$ is generated, its vehicle's geographical location determines the set of network rings covering the vehicle and the corresponding channel gains. Let $\mathcal{R}_j$ be the set of network rings covering $n_j$'s vehicle at the time of its release, where $|\mathcal{R}_j| \leq M$. Based on the vehicle's moving speed and direction, the duration $t_{jkr}^v$ for which the vehicle remains within the coverage of network ring $m_{kr} \in \mathcal{R}_j$ can be estimated. If $m_{kr}$ is selected for offloading $n_j$, the processing of $n_j$ must be completed before $n_j$'s vehicle leaves the coverage of $m_{kr}$. Furthermore, if $n_j$ is processed on $m_k$ using $c$ computing units, let $t_{jkc}^p$ represent the processing time for $n_j$.

\begin{definition}[Assignment Instance]\label{def:instance}
    An assignment instance of job $n_j$ is defined as $\ell \triangleq \langle m_{\ell r}, b_{\ell}, c_{\ell} \rangle$, which represents the scenario where $n_j$ is offloaded to ES $m_{\ell}$ when $n_j$'s vehicle is covered by network ring $m_{\ell r} \in \mathcal{R}_j$ and $n_j$ is allocated $b_\ell$ bandwidth units and $c_\ell$ computing units. 
\end{definition}

Suppose job $n_j$ is assigned based on assignment instance $\ell \triangleq \langle m_{\ell r}, b_{\ell}, c_{\ell} \rangle$ (Definition \ref{def:instance}). Based on Shannon's theorem \cite{shannon1984}, the data offloading rate is given by $\eta_{j} = b_\ell \cdot  \beta_{\ell}  \log_2(1+p_{j} \cdot  h_{\ell r}/{\sigma^2})$,
where $b_\ell \cdot  \beta_{\ell}$ is the allocated bandwidth, $p_j$ is the offloading power of $n_j$'s vehicle, $h_{\ell r}$ is the channel gain in network ring $m_{\ell r}$, and $\sigma$ is the noise spectral density. The offloading time $t_{j}^o$ is given by $t_{j}^o = \theta_j / \eta_{j}$.
Given computation resource allocation $c_\ell$, the processing time of $n_j$ is $t_j^p = d_{j\ell c_\ell}^p$. Given that the data size of results is typically much smaller than job inputs, and the bandwidth capacity of the downlink is generally larger than that of the uplink in 5G networks, we disregard the result return time, as done in \cite{lin2023a}. Let $t_j = t_{j}^o + t_j^p$ denote the completion time of $n_j$. Each offloaded job $n_j$ offers a utility $u_j$ depending on $t_j$, which is defined as follows.
\begin{equation}\label{eq:sys2}
    u_j(t_j) = 
    \begin{cases}
        U_j & \text{if } t_j \le \Delta_j\\
        U_j \cdot \Psi_j(t_j) & \text{if }  \Delta_j < t_j \le \gamma_j\Delta_j\\
        0 & \text{otherwise}
    \end{cases}.
\end{equation}
$\gamma_j \ge 1$ represents the tolerance factor for deadline violations of job $n_j$. The full utility offered by $n_j$ when completed before its deadline $\Delta_j$ is denoted by $U_j$. Let $u(\ell)$ be the utility obtained from $n_j$ when it is assigned according to $\ell$, as defined in Eq.~\eqref{eq:sys2}. $\Psi_j$ is a utility-reduction function for deadline violations that maps to the range $[0, 1]$. This function $\Psi_j$ can take various forms (e.g., linear, nonlinear, or step functions), ensuring that $u_j(t_j)$ is flexible enough to model diverse practical scenarios. For example, $u_j(t_j)$ may represent a function related to energy consumption, completion time, or a combination of both factors.

In IoV, safety-critical jobs (e.g., object detection for autonomous driving) typically have hard deadlines that do not permit any deadline violations. Conversely, non-safety-critical jobs (e.g., image or video rendering) have soft deadlines, tolerating some deadline violations at the cost of reduced user experience when deadlines are missed. Therefore, this work considers jobs with both hard and soft deadlines. If $n_j$ has a hard deadline, then $\gamma_j = 1$. If $n_j$ has a soft deadline, $\gamma_j > 1$, and the utility derived from $n_j$ decreases from $U_j$ to $0$ according to $\Psi_j$ after its deadline is exceeded.

\subsection{Problem Formulation}
In this paper, we aim to develop a job offloading and resource allocation strategy that maximizes the utility gained from jobs while satisfying system resource constraints, ES access constraints, and job deadline constraints (Problem $\mathbf{P}$). 
Next, we formulate $\mathbf{P}$ in its generalized form.



\begin{definition}[Feasible Assignment Instance]\label{def:feasible_instance}
An assignment instance $\ell \triangleq \langle m_{\ell r}, b_{\ell}, c_{\ell} \rangle$ of job $n_j$ is feasible if $\ell$ satisfies (i) the ES access constraint (i.e., $m_{\ell r} \in \mathcal{R}_j$), (ii) the deadline constraint (i.e., $t_j \le \min\{\gamma_j\Delta_j, t_{j\ell r}^v\}$), and (iii) the resource constraint (i.e., $b_{\ell} \le B_\ell, c_{\ell} \le C_\ell$).
\end{definition}

\begin{table}
    \caption{Notation (Main Parameters and Variables)}
    \resizebox{\columnwidth}{!}{
    \begin{tabular}[t]{|m{0.04\textwidth}|p{0.4\textwidth}|}
        \hline
        Symb. & Definition \\ \hline
        $\mathcal{M}$ & the set of $M$ ESs; $m_k \in \mathcal{M}$ denotes an ES \\ \hline
        $C_k$ & the total computing units of ES $m_k$; $C_k \le C$ \\ \hline
        $B_k$ & the total bandwidth units of ES $m_k$; $B_k \le B$ \\ \hline
        $m_{kr}$ & $r$-th network ring of ES $m_k$ \\ \hline
        $\mathcal{N}$ & the set of $N$ jobs; $n_j \in \mathcal{N}$ denotes a job \\ \hline
        $\mathcal{R}_j$ & the set of accessible network rings of job $n_j$; $|\mathcal{R}_j| \le M$ \\ \hline
        $\Delta_j$ & deadline of job $n_j$ \\ \hline
        $\gamma_j$ & tolerance factor for deadline violation of job $n_j$; $\gamma_j \ge 1$ \\ \hline
        $U_j$ & utility gained from job $n_j$ when completed before $\Delta_j$ \\ \hline
        $t_j$ & completion time of job $n_j$   \\ \hline
        $\ell$ & $\ell \triangleq \langle m_{\ell r}, b_{\ell}, c_{\ell} \rangle$ is an assignment instance of a job \\ \hline
        $u(\ell)$ & utility gained from job $n_j$ when assigned based on $\ell$ \\ \hline
        $\mathcal{L}$ & the set of all assignment instances of all jobs \\ \hline
        $\mathcal{L}_j$ & the set of all assignment instances of job $n_j$ \\ \hline
        $\mathcal{E}_k$ & the set of all assignment instances mapped to ES $m_k$ \\ \hline
        $\mathcal{L}(\ell)$ & set of assignment instances related to the same job as $\ell$ \\ \hline
        $\mathcal{E}(\ell)$ & set of assignment instances mapped to the same ES as $\ell$ \\ \hline
        $\mathcal{L}_L$ & the set of all light assignment instances in $\mathcal{L}$ \\ \hline
        $\mathcal{L}_H$ & the set of all heavy assignment instances in $\mathcal{L}$ \\ \hline
        \hline 
        $x_\ell$ & binary selection variable of assignment instance $\ell$ \\ \hline
    \end{tabular}
    }
    \label{table:notation}
\end{table}

Given the finite bandwidth and computing resources of each ES, for each job $n_j \in \mathcal{N}$, we enumerate all its potential assignment instances by considering each $m_{kr} \in \mathcal{R}_j$ and all possible resource allocations. We then restrict our focus to feasible assignment instances (as defined in Definition \ref{def:feasible_instance}). Henceforth, \textit{unless explicitly stated otherwise, any reference to an assignment instance refers to a feasible assignment instance}. Let $\mathcal{L}$ denote the set of all feasible assignment instances for all jobs, $\mathcal{L}_j \subset \mathcal{L}$ the set of assignment instances for job $n_j$, and $\mathcal{E}_k \subset \mathcal{L}$ the set of assignment instances associated with ES $m_k$. For a given assignment instance $\ell$, let $\mathcal{L}(\ell)$ represent the set of all assignment instances corresponding to the same job as $\ell$, and let $\mathcal{E}(\ell)$ denote the set of all assignment instances associated with the same ES as $\ell$. Since each job has at most $M$ accessible network rings, $B$ bandwidth allocation choices, and $C$ computation resource options, the total number of assignment instances in $\mathcal{L}$ is at most $NMBC$. For any set of assignment instances $\mathcal{S}$, let $u(\mathcal{S}) = \sum_{\ell \in \mathcal{S}} u(\ell)$. For each assignment instance $\ell \in \mathcal{L}$, let $x_\ell \in {0, 1}$ be the selection variable of $\ell$, where $x_\ell = 1$ if and only if $\ell$ is chosen in the solution. Then, we formulate $\mathbf{P}$ as follows.
\begin{subequations}
    \begin{align}
        (\mathbf{P}) \ \max  {\textstyle \sum_{\ell \in \mathcal{L}}} u(\ell) & \cdot x_\ell \label{eq:eop}\\
        \text{subject to:} \qquad {\textstyle \sum_{\ell \in \mathcal{E}_k} b_\ell \cdot x_\ell \le B_k},& \forall m_k \in \mathcal{M}\label{eq:eop1} \\
        {\textstyle \sum_{\ell \in \mathcal{E}_k} c_\ell \cdot x_\ell \le C_k},& \forall m_k \in \mathcal{M}\label{eq:eop2} \\
        {\textstyle \sum_{\ell \in \mathcal{L}_j} x_\ell \le 1},& \forall n_j \in \mathcal{J}\label{eq:eop3} \\
        x_\ell \in \{0,1\},& \forall \ell \in \mathcal{L} \label{eq:eop4}
    \end{align}
Eqs. \eqref{eq:eop1} and \eqref{eq:eop2} are the bandwidth and computation resource constraints for each ES, respectively. 
\end{subequations}
Eq. \eqref{eq:eop3} ensures that at most one assignment instance is selected for each job. Notably, since set $\mathcal{L}$ consists of only feasible assignment instances, the deadline and ES access constraints have been implicitly satisfied.

\section{Instance Dividing Assignment Algorithm}
In this section, we present an approximation algorithm for $\mathbf{P}$, named the Instance Dividing Assignment Algorithm ($\mathtt{IDAssign}$), which leverages a local ratio framework combined with an instance-dividing technique. Later in this section, we formally prove that $\mathtt{IDAssign}$ achieves a $\frac{1}{6}$-approximation ratio for $\mathbf{P}$. To the best of our knowledge, $\mathtt{IDAssign}$ is the first algorithm that provides a constant approximation ratio for DOAP.

The detailed steps of $\mathtt{IDAssign}$ are outlined in Algorithm \ref{alg:divide}. For each assignment instance $\ell$, let $\tilde{b}_\ell = \frac{b_\ell}{B_\ell}$ and $\tilde{c}_\ell = \frac{c_\ell}{C_\ell}$ represent the normalized bandwidth and computation resource allocation, respectively. We classify an assignment instance $\ell$ as a \textit{light instance} if $\tilde{b}_\ell \le \frac{1}{2}$ and $\tilde{c}_\ell \le \frac{1}{2}$, and as a \textit{heavy instance} otherwise. Let $\mathcal{L}_L$ denote the set of all light instances in $\mathcal{L}$, and $\mathcal{L}_H$ the set of all heavy instances in $\mathcal{L}$. The algorithm is then recursively called with the initial invocation $\mathtt{IDAssign}(\mathbf{u}, \mathcal{L})$, where $\mathbf{u}$ is a vector containing the utility values $u(\ell)$ for all $\ell \in \mathcal{L}$. Since the utility of each assignment instance is updated at each recursive layer, we denote the updated utility of $\ell$ at the $i$-th recursive layer as $w^i(\ell)$, where $u(\ell)$ represents the initial utility of $\ell$. 


In each recursive layer $i$, the algorithm selects an instance $\ell^i$ with the smallest $\max\{\tilde{b}_{\ell^i}, \tilde{c}_{\ell^i}\}$, prioritizing instances in $\mathcal{L}_L$ over those in $\mathcal{L}_H$ (line $4$). The utility value $w^i(\ell)$ for each assignment instance $\ell \in \mathcal{L}$ is then decomposed into two components, $w_1^i(\ell)$ and $w_2^i(\ell)$ (line $5$), where $w_2^i(\ell)$ represents the marginal gain corresponding to the choice of $\ell^i$. The marginal utility vector $\mathbf{w_2^i}$ becomes the utility vector $\mathbf{w^{i+1}}$ for the next recursive layer (line $6$). Based on the solution $\mathcal{S}^{i+1}$ returned from layer $(i+1)$, $\ell^i$ is added to the solution if it satisfies the feasibility constraints \eqref{eq:eop1}--\eqref{eq:eop3}. Next, we prove the approximation ratio of $\mathtt{IDAssign}$ for $\mathbf{P}$ as follows.

\begin{algorithm}[tb]
    \caption{\scalebox{1}{Instance Dividing Assignment Algorithm}}\label{alg:divide}

    \SetAlgoHangIndent{0pt} 
    \SetKwProg{Fn}{}{:}{}
    \SetKwInOut{Input}{input}
    \SetKwInOut{Output}{output}
    \SetKw{Return}{return}
    
    \Fn{$\mathtt{IDAssign}$($\mathbf{w^i}$, $\mathcal{L}$)}{
    Remove all instances $\ell$ from $\mathcal{L}$ with $w^i(\ell) \le 0$\;
    \lIf{$\mathcal{L} = \emptyset$}{\Return $\mathcal{S}^i=\emptyset$}
    \leIf{$\mathcal{L}_L \neq \emptyset$}{$\ell^i \leftarrow \argmin_{\ell \in \mathcal{L}_L} (\max\{\tilde{b}_\ell, \tilde{c}_\ell\})$}{$\ell^i \leftarrow \argmin_{\ell \in \mathcal{L}_H} (\max\{\tilde{b}_\ell, \tilde{c}_\ell\})$}
    
    Decompose utility vector $\mathbf{w^i}$ into $\mathbf{w_1^i}$ and $\mathbf{w_2^i}$ (i.e., $\mathbf{w^i}=\mathbf{w_1^i}+\mathbf{w_2^i}$) such that for any $\ell \in \mathcal{L}$,
    \begin{equation*}
         w_1^i(\ell) = w^i(\ell^i) \cdot
        \begin{cases}
            1 & \text{if } \ell \in \mathcal{L}(\ell^i)\\
            (\tilde{b}_{\ell} + \tilde{c}_{\ell}) & \text{if }  \ell \in \mathcal{E}(\ell^i) \setminus \mathcal{L}(\ell^i)\\
            0 & \text{otherwise}
        \end{cases};
    \end{equation*}
    
    $\mathcal{S}^{i+1} \leftarrow $ $\mathtt{IDAssign}$($\mathbf{w_2^i}$, $\mathcal{L}$)\;
    \leIf{$\mathcal{S}^{i+1} \cup \{\ell^i\}$ is feasible}{\Return $\mathcal{S}^i \leftarrow \mathcal{S}^{i+1} \cup \{\ell^i\}$}{\Return $\mathcal{S}^i \leftarrow \mathcal{S}^{i+1}$}
    }
\end{algorithm}

\begin{theorem}\label{theorem:divide}
    \scalebox{0.97}{$\mathtt{IDAssign}$ is a $\frac{1}{6}$-approximation algorithm for $\mathbf{P}$.}
\end{theorem}
\begin{proof}
    Suppose $\mathcal{Q}$ is an optimal solution of $\mathbf{P}$.
    Let the innermost recursive layer of $\mathtt{IDAssign}$ be layer $I$, and the outermost be layer $1$. We use simple induction to prove this lemma by showing that $w^i(\mathcal{S}^i) \ge \frac{1}{6}w^i(\mathcal{Q})$ for $i=I, I-1, ..., 1$ (since $\mathbf{w^1} = \mathbf{u}$).

    \textbf{Base Case} ($i=I$): When $i=I$, $w^i(\ell) \le 0, \forall \ell \in \mathcal{L}$. The returned $\mathcal{S}^I = \emptyset$, so 
    $w^I(\mathcal{S}^I) = 0 \ge w^I(\mathcal{Q}) \ge \frac{1}{6}w^I(\mathcal{Q})$.

    \textbf{Inductive Step} ($i < I$): When $i < I$, suppose $w^{i+1}(\mathcal{S}^{i+1}) \ge \frac{1}{6}w^{i+1}(\mathcal{Q})$. Since $\mathbf{w_2^i} = \mathbf{w^{i+1}}$, $w_2^{i}(\mathcal{S}^{i+1}) \ge \frac{1}{6}w_2^{i}(\mathcal{Q})$. Based on the utility decomposition in line $5$, $w_1^i(\ell^i)=w^i(\ell^i)$, so $w_2^i(\ell^i)$$=$$0$; thus,
    \begin{equation}\label{eq:light1}
        {\textstyle w_2^{i}(\mathcal{S}^{i}) = w_2^{i}(\mathcal{S}^{i+1}) \ge \frac{1}{6}w_2^{i}(\mathcal{Q})}.
    \end{equation}
    In line $7$, if $\ell^i$ can be added to $\mathcal{S}^{i}$, $w_1^{i}(\mathcal{S}^{i}) \ge w^i(\ell^i)$. 
    Otherwise, either one $\ell \in \mathcal{L}(\ell^i)$ already exists in $\mathcal{S}^{i+1}$ that stops $\ell^i$ from being added to $\mathcal{S}^{i}$ (constraint \eqref{eq:eop3}) or the remaining resource in ES $m_{\ell^i}$ is not enough to accommodate $\ell^i$ (constraints \eqref{eq:eop1} and \eqref{eq:eop2}). In the former case, $w_1^{i}(\mathcal{S}^{i}) \ge w^i(\ell^i)$ due to line $5$. In the latter case, let $\mathcal{A}^{i} = \mathcal{S}^{i+1} \cap \mathcal{E}(\ell^i)$, representing the instances that have been selected in $\mathcal{S}^{i+1}$ and have resource contention with $\ell^i$. We consider the following two situations for the latter case: $\ell^i \in \mathcal{L}_H$ and $\ell^i \in \mathcal{L}_L$. If $\ell^i \in \mathcal{L}_H$, as we only consider instances in $\mathcal{L}_H$ when $\mathcal{L}_L=\emptyset$, the solution $\mathcal{S}^{i+1}$ contains only heavy instances; in this situation, there exists at least one heavy instance $\ell \in \mathcal{S}^{i+1}$ that prevents $\ell^i$ being added to the solution due to resource contention. Based on line $5$ and the definition of heavy instances, $w_1^{i}(\mathcal{S}^{i}) \ge \frac{1}{2}w^i(\ell^i)$. If $\ell^i \in \mathcal{L}_L$, either $\sum_{\ell \in \mathcal{A}^{i}}\tilde{b}_{\ell} \ge \frac{1}{2}$ or $\sum_{\ell \in \mathcal{A}^{i}}\tilde{c}_{\ell} \ge \frac{1}{2}$ (i.e., at least $\frac{1}{2}$ of $m_{\ell^i}$'s bandwidth or computation resource are allocated to $ \mathcal{S}^{i+1}$); thus, $w_1^{i}(\mathcal{S}^{i})  \ge \frac{1}{2}w^i(\ell^i)$ due to line $5$. 
    Therefore, if instance $\ell^i$ cannot be added to $\mathcal{S}^{i}$, we have $w_1^{i}(\mathcal{S}^{i}) \ge \frac{1}{2}w^i(\ell^i)$. 
    Besides, $w_1^i(\mathcal{Q})$ is maximized when $\mathcal{Q}$ contains one $ \ell \in \mathcal{L}(\ell^i) \setminus \mathcal{E}(\ell^i)$ and all resources of ES $m_{\ell^i}$ has been allocated to $\ell' \in \mathcal{Q} \cap \mathcal{E}(\ell^i) \setminus \mathcal{L}(\ell^i)$; in this case, $w_1^i(\mathcal{Q}) = w^i(\ell^i) + 2w^i(\ell^i) = 3w^i(\ell^i).$
    Hence, we have 
    \begin{equation}\label{eq:light2}
        w_1^{i}(\mathcal{S}^{i}) \ge \frac{1}{6}w_1^i(\mathcal{Q}).
    \end{equation}
    Given Eq. \eqref{eq:light1}, Eq. \eqref{eq:light2}, and $\mathbf{w^i}=\mathbf{w_1^i}+\mathbf{w_2^i}$, we have $w^{i}(\mathcal{S}^{i}) \ge \frac{1}{6}w^i(\mathcal{Q})$. Based on simple induction, $w^{i}(\mathcal{S}^{i}) \ge \frac{1}{6}w^i(\mathcal{Q})$ for $i = I, ..., 1$.
\end{proof}

\textit{Time Complexity.} Since $\mathcal{L}$ contains at most $NMBC$ assignment instances, $\mathtt{IDAssign}$ has at most $NMBC$ recursive layers. In line $5$, at most $NMBC$ assignment instances are considered for utility decomposition. Thus, the time complexity of $\mathtt{IDAssign}$ is $\mathcal{O}((NMBC)^2)$.

\begin{figure*}
    \centering
    \subfigure[]{\includegraphics[width=0.3\textwidth, page=3]{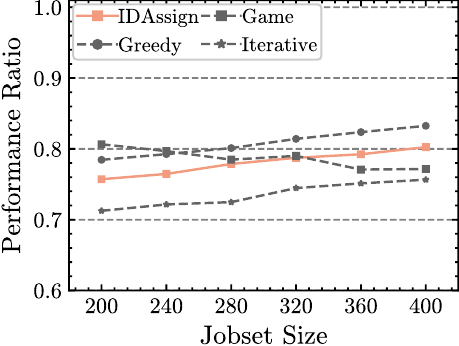} \label{fig:exp2a}} 
    \subfigure[]{\includegraphics[width=0.3\textwidth, page=4]{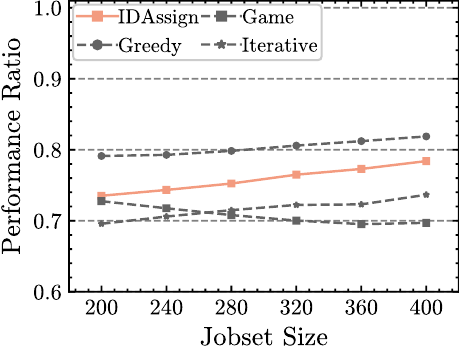} \label{fig:exp2b}}
    \subfigure[]{\includegraphics[width=0.3\textwidth, page=2]{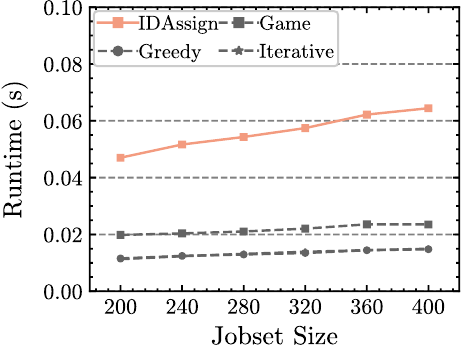} \label{fig:exp2c}}
    \caption{(a) average performance ratios of algorithms under the (low,low) range combination for different jobset sizes; (b) average performance ratios of algorithms under the (high,high) range combination for different jobset sizes; (c) average runtime of algorithms for different jobset sizes (the plots for $\mathtt{Greedy}$ and $\mathtt{Iterative}$ are overlapped);}
    \Description{experiments for mec with 20 edge server}
\end{figure*}

\section{Numerical Evaluation}\label{sec:experiment}
In this section, we evaluate the performance of $\mathtt{IDAssign}$. This algorithm is compared against three existing algorithms from the literature, focusing on utility performance and runtime efficiency. All experiments were performed on a desktop PC equipped with an Intel(R) Core(TM) i7-14700KF CPU and 64 GB of RAM.


\subsection{Simulation Setup}
A real taxi GPS data trace \cite{taxi-data}, collected on April 1, 2018 in Shanghai by Shanghai Qiangsheng Taxi Company, is used as the vehicle traffic data. We select a city area of $1$km$\times1$km in Shanghai and extract all taxi trajectories passing through this area over a period of $15$ minutes during the evening peak. We consider a MEC with $20$ ES evenly distributed along the roads. Each ES is mounted at a height of $10$m, and its network is divided into $2$ rings with coverage ranges of $0\sim100$m and $100\sim200$m. 
For each ES $m_k$, $\beta_k$ is set to $2$MHz (the smallest BU type in WiFi-6), and $B_k$ is set to either $20$ or $40$. The data offloading rate $\eta_j$ for one BU is set to $1.65$ MBps for ring $1$ and $1.15$ MBps for ring $2$. Each ES is equipped with a GPU randomly sampled from $5$ types of GPUs (TITAN XP, TITAN RTX, RTX3090, RTX4060Ti, and A100). We set $C_k$ of each ES $m_k$ as $25$. 

In this experiment, we synthesize jobsets with varying resource utilizations and jobset sizes. The bandwidth utilization $ru_b$ of a jobset is defined as the ratio of the total bandwidth demand of all jobs in a jobset to the total bandwidth available in the MEC \cite{gao2022deadline}, and the computation resource utilization $ru_c$ of a jobset is analogously defined. We consider two resource utilization sampling ranges: \textbf{low} ($[0.6, 0.9]$) and \textbf{high} ($[1.2, 1.5]$), forming four different range combinations. Jobset size $N$ ranges from $200$ to $400$ in steps of $40$. For each jobset size, we sample $150$ different $(ru_b, ru_c)$ pairs from each range combination. In total, we synthesize $3600$ jobsets.

Given a jobset size and the sampled $(ru_b, ru_c)$, we synthesize a jobset as follows. The jobset release time is randomly sampled within the $15$-minute period, and each job is randomly mapped to a taxi active at the release time to obtain the ring coverage $\mathcal{R}_j$. For each job $n_j$, we randomly sample its input $\theta_j$ from $45$ images ranging from $0.15$ to $0.63$ MB, and its utility $U_j$ from the range $[20, 60]$. Then, we randomly sample the application type of $n_j$ from 9 GPU applications (resnet34/50/101, densenet121/169, and vgg11/13/16/19) for object detection.
We use Stafford's Randfixedsum Algorithm~\cite{emberson2010techniques} to distribute $ru_b$ and $ru_c$ to each individual job in the jobset in a uniformly random and unbiased manner, which is used to compute the job deadline.
Next, we compute the maximum tolerable completion time $\gamma_j\Delta_j$ where we uniformly sample $\gamma_j$ from the range $[1.8, 2.2]$ and define the function $\Psi_j$ as a linear decreasing function when $n_j$ has a soft deadline.

\textit{Baseline Algorithms}. We compare $\mathtt{IDAssign}$ with three heuristic algorithms that are summarized from the literature on DOAP: $\mathtt{Greedy}$ \cite{gao2022deadline}, $\mathtt{Iterative}$ \cite{hu2020dynamic, li2019cooperative, qi2024joint, fan2023joint}, and $\mathtt{Game}$ \cite{li2019cooperative}. Note that none of these baselines provide theoretical guarantees on the performance of their algorithms for problem $\mathbf{P}$. $\mathtt{Greedy}$ defines $e_\ell = u(\ell)/(\frac{b_{\ell}}{B_\ell}\times \frac{c_{\ell}}{C_\ell})$ as the resource efficiency of $\ell \in \mathcal{L}$, and adds $\ell \in \mathcal{L}$ into the solution whenever feasible in a descending order of $e_\ell$. $\mathtt{Iterative}$ decouples $\mathbf{P}$ into a job offloading subproblem and a resource allocation subproblem; then, solve these subproblems iteratively until convergence. $\mathtt{Game}$ defines $\mathbf{P}$ as a non-cooperative game, where each vehicle is a player; in each round, select the $\ell \in \mathcal{L}$ that contributes the largest increment on the total utility gained.

\textit{Metrics}. We use Performance Ratio to measure the performance of all algorithms, defined as the ratio of the total utility obtained by an algorithm to the optimal utility for $\mathbf{P}$ (the optimal utility is obtained by solving $\mathbf{P}$ using the ILP solver Gurobi, with a timeout limit set to 10 minutes).

\subsection{Result Discussion}
We first evaluated the average performance ratios of the algorithms across various jobset sizes under the range combinations (low, low) and (high, high) (Figs. \ref{fig:exp2a} and \ref{fig:exp2b}). The combination (low, low) indicates that both $ru_b$ and $ru_c$ of the jobsets are sampled from the \textbf{low} range. For both range combinations, $\mathtt{IDAssign}$ demonstrates improved performance as the jobset size increases, with better results observed under the (low, low) combination. This performance trend occurs because $\mathtt{IDAssign}$ prioritizes light instances (i.e., $\mathcal{L}_L$) over heavy instances (i.e., $\mathcal{L}_H$). When the jobset size is small or when the resource demand falls within the (high, high) range combination, the per-job resource demand increases, making heavy instances more favorable candidates in an optimal solution. However, due to the prioritization strategy in $\mathtt{IDAssign}$, the utility values $w^i(\ell)$ of many heavy instances $\ell \in \mathcal{L}_H$ are frequently reduced to zero before being considered, thereby decreasing the likelihood of selecting heavy instances and ultimately diminishing the overall performance of $\mathtt{IDAssign}$.

The average runtime of algorithms for different jobset sizes is illustrated in Fig. \ref{fig:exp2c}. While offering a theoretical guarantee and competitive practical performance, $\mathtt{IDAssign}$ achieves a runtime comparable to the baseline algorithms ($\mathtt{Greedy}$, $\mathtt{Game}$, and $\mathtt{Iterative}$). 
\textit{This makes $\mathtt{IDAssign}$ a preferable choice when both a theoretical performance guarantee and minimal computational overhead are required.} Furthermore, $\mathtt{IDAssign}$ exhibits nearly linear runtime scalability with respect to jobset size, confirming its practical suitability for handling large-scale jobsets.

\begin{figure}
    \centering
    \includegraphics[width=0.9\linewidth, page=1]{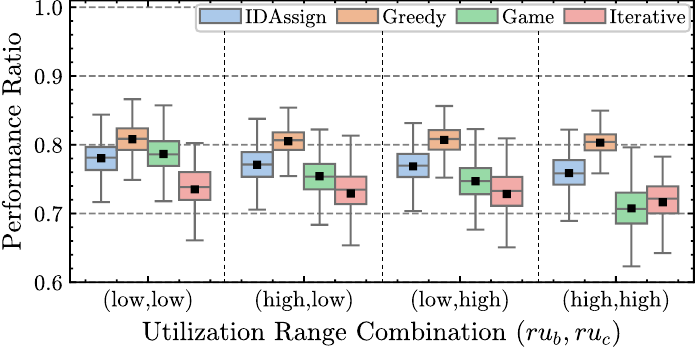}
    \caption{Performance ratios of algorithms under different resource utilization range combinations ($ru_b, ru_c$)}
    \label{fig:exp1}
    \Description{performances of all algorithms under different resource utilization}
\end{figure}

The overall algorithm performance across all resource utilization range combinations (including (low,low), (low, high), (high, low) and (high,high)) in the MEC with $20$ ES is presented in Fig. \ref{fig:exp1}. $\mathtt{IDAssign}$ achieves an average performance ratio of $76.9\%$. Despite the \textit{baseline algorithms lacking theoretical guarantees}, $\mathtt{IDAssign}$ still demonstrates a practical performance comparable to $\mathtt{Greedy}$, $\mathtt{Game}$, and $\mathtt{Iterative}$, underscoring its practical effectiveness. Remarkably, $\mathtt{IDAssign}$ surpasses its theoretical bound of $\frac{1}{6}$ by $60.2\%$, illustrating its ability to narrow the gap between theoretical guarantees and optimal solutions. 

\begin{figure}
    \centering
    \includegraphics[width=0.6\columnwidth, page=5]{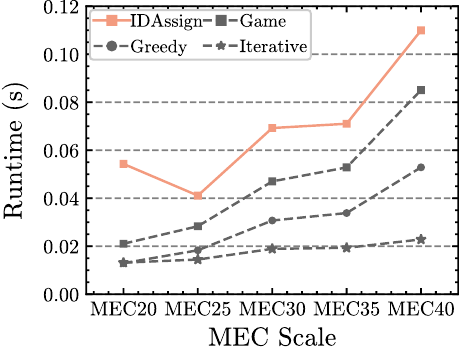}
    \caption{average runtime of algorithms for different MEC scales (with jobset size $280$)}
    \label{fig:exp3a}
    \Description{runtime for different MEC scales}
\end{figure}

To evaluate the runtime scalability of $\mathtt{IDAssign}$ across different MEC scales, we consider four additional MEC with $25$, $30$, $35$, and $40$ ES, referred to as MEC$25$, MEC$30$, MEC$35$, and MEC$40$, respectively. The base configuration with $20$ ES is denoted as MEC$20$. For each MEC, the jobset size is fixed at $280$, and $150$ jobsets are synthesized for each resource utilization range combination. The runtime performance across varying MEC scales is shown in Fig. \ref{fig:exp3a}. With a constant jobset size, the runtime of $\mathtt{IDAssign}$ exhibits linear growth with the number of ES. This scalability trend, together with observations from Fig. \ref{fig:exp2c}, demonstrates that $\mathtt{IDAssign}$ achieves efficient runtime scalability with respect to both jobset size and the number of ES in the MEC.


\section{Conclusion}
This paper addressed a joint job offloading and resource allocation problem, $\mathbf{P}$, in MEC with both bandwidth and computation resource constraints, aiming to maximize the total utility gained from jobs. We proposed an approximation algorithm, $\mathtt{IDAssign}$, for $\mathbf{P}$ with a theoretical bound of $\frac{1}{6}$. $\mathtt{IDAssign}$ provided the first constant approximation bound for $\mathbf{P}$. Experimental results showed that $\mathtt{IDAssign}$ delivered superior practical performance and exhibited strong runtime scalability. Notably, despite the absence of theoretical guarantees for the baseline algorithms, $\mathtt{IDAssign}$ achieves comparable practical performance, underscoring its efficiency and effectiveness.
In the future, we plan to explore the online deployment of resource allocation algorithms in MEC environments with dynamic network conditions.

\begin{acks}
This work was supported in part by the MoE Tier-2 grant MOE-T2EP20221-0006.
\end{acks}


\bibliographystyle{ACM-Reference-Format}


\end{document}